%% file: main.tex
\documentclass[conference]{IEEEtran}
\usepackage[T1]{fontenc}
\usepackage{amsmath,amssymb}
\usepackage{graphicx}
\usepackage{booktabs}
\usepackage{xcolor}
\usepackage{tikz}
\usetikzlibrary{arrows.meta,positioning,fit,backgrounds}
\usepackage[hidelinks]{hyperref}
\usepackage{cite}
\newtheorem{proposition}{Proposition}

\input{generated/numbers.tex}

\begin{document}

\title{Fast Path, Exact Path: One Library of Hardwired Operators for Streaming and Exact Canny Edge Detection on a Low-Cost FPGA}

\author{\IEEEauthorblockN{Raj Mali}
\IEEEauthorblockA{\textit{Independent Researcher}\\
{raj.mali.research@gmail.com}\\
ORCID: \href{https://orcid.org/0009-0001-6351-0062}{0009-0001-6351-0062}}}

\maketitle

\begin{abstract}
Hysteresis edge tracking, the final stage of Canny edge detection, follows edge chains of unbounded length and therefore resists single-pass streaming; FPGA designs either approximate it or dedicate separate labelling hardware to it. We show that a single library of hardwired $3\times3$ operators, each with an identical two-clock latency, delivers both a camera-rate approximation and exact hysteresis on a low-cost Artix-7 (XC7A35T). The \emph{fast path} streams camera pixels through the operators and $K$ local hysteresis stages and commits each edge decision \latKOne--\latKEight~$\mu$s after the pixel arrives ($K=1$--$8$; system simulation with a datasheet camera model), that is, $3+K$ camera lines later and \soonerKEight--\soonerKOne$\times$ sooner than a frame-based processor can even begin at 30 frames per second. It provably never reports an edge that exact hysteresis rejects, recovers \recallPhotoKOne--\recallPhotoKEight\,\% of exact edges at thresholds 20/40, and on the \bsdsN{} BSDS500 test images retains at least \bsdsFastMinPctOfExact\,\% of the exact path's boundary F-measure (ODS \bsdsFastMinF--\bsdsFastMaxF{} versus \bsdsExactF). The \emph{exact path} holds $128\times128$ tiles in block RAM, is exact within each tile, and replaces two-buffer iteration with in-place forward and reverse sweeps built into the same operator. This reduces hysteresis from \structTwoBufMin--\structTwoBufMax{} passes to \structInplaceMin--\structInplaceMax{} on pattern and photo tiles, accelerating the hysteresis step \hystSpeedMin--\hystSpeedMax$\times$ (\hystSpeedWorst$\times$ on a worst-case tile) and a complete tile \tileSpeedMin--\tileSpeedMax$\times$, with identical output. Both paths are bit-exact with software golden models at every pipeline stage and close timing at 50~MHz in at most \fastLut{} LUTs. An optimised multi-core CPU remains faster on stored images, but it cannot begin until the frame is complete.
\end{abstract}


\section{Introduction}
Edge detection is the entry point of many vision systems, from inspection lines to robotic grasping, and Canny's detector~\cite{canny1986} remains the reference because it yields thin, connected edges with few false responses. In embedded and industrial settings, however, the decisive figure is often not throughput but latency: \emph{when} a decision becomes available, and how predictably.

Four of Canny's five stages, namely Gaussian smoothing, Sobel gradients~\cite{sobel1968}, non-maximum suppression (NMS) and double thresholding, are local $3\times3$ operations that map naturally onto a line-buffered FPGA pipeline. The fifth, hysteresis, is not: a weak pixel is an edge only if a chain of weak pixels links it to a strong one, however long that chain. Software follows chains with a stack~\cite{bradski2000,moore2011}, which presupposes random access. FPGA designs therefore settle for one of three compromises. They approximate hysteresis within a fixed window, which is fast but inexact; they add a connected-component labelling unit, which is exact but separate hardware~\cite{alnajjar2010,jang2019}; or they store the frame and iterate a local rule to convergence, which is exact and reuses the local operator but is slow, because a chain grows by only one pixel per full-frame pass (up to \twoBufRealMax{} passes on a $128\times128$ photo tile in our design).

CPU comparisons demand equal rigour, since weak baselines have inflated many reported speed-ups~\cite{lee2010}. An optimised library on a modern CPU processes a stored $480\times720$ frame in \cvAllMed~ms (Section~\ref{sec:cpu}), which no one-pixel-per-clock pipeline at 50~MHz can match. The FPGA's advantage lies elsewhere: it can decide before the frame is complete, at a latency fixed by its structure.

We argue that speed and exactness need not be engineered separately. A single library of hardwired operators serves two architectures: a \emph{fast path} that decides within a few image lines and never adds a false edge, and an \emph{exact path} that completes the weak chains the fast path may miss. A system can act on the fast answer and refine it with the exact one, much as fast reflexes and slower deliberate processing coexist in nervous systems. Our contributions are:
\begin{enumerate}
\item A library of $3\times3$ Canny operators with a common two-clock latency that serves, unchanged except for one feedback-enable input, both a streaming fast path and a frame-buffered exact path (Fig.~\ref{fig:arch}, Section~\ref{sec:arch}).
\item In-place forward/reverse hysteresis sweeps, a known technique for morphological reconstruction~\cite{vincent1993}, built into the shared operator through a one-clock left-neighbour feedback and a row-above write-back into the window's line buffer. The output is identical to two-buffer iteration, while the hysteresis step runs \hystSpeedMin--\hystSpeedMax$\times$ faster on photo and pattern tiles (Section~\ref{sec:inplace}).
\item A verification methodology in which every pipeline stage of both paths is compared byte for byte with bit-exact software models in RTL simulation, and on hardware for the exact path (Section~\ref{sec:verif}).
\item Measurements of fast-path latency and accuracy as functions of $K$, including boundary accuracy on BSDS500, and a fair CPU comparison that reports where the CPU wins (Section~\ref{sec:results}).
\end{enumerate}

\section{Background and Related Work}
\textbf{Hysteresis in software.} Canny's hysteresis retains a weak pixel only if it is 8-connected through weak pixels to a strong one~\cite{canny1986}. OpenCV traces chains from each strong pixel with a stack, visiting each pixel roughly once~\cite{bradski2000}, and a patent application describes the same seed-and-stack scheme~\cite{moore2011}. The random access this requires suits a CPU, not a raster-order pipeline.

\textbf{Canny on FPGAs.} Early FPGA implementations stream the local stages through line buffers~\cite{neoh2004}. Xu et al.\ partition the frame into blocks processed in parallel with block-adaptive thresholds~\cite{xu2014}, and recent designs on low-cost devices report video-rate operation with modified Canny variants~\cite{guo2023}. Other implementations parallelise the pipeline for real-time operation~\cite{gentsos2010}, reduce its hardware cost~\cite{sangeetha2019}, adapt the Canny parameters to the estimated noise level~\cite{kalbasi2020}, or combine edge and corner detection at several resolutions~\cite{possa2014}. Stencil compilers such as Darkroom~\cite{hegarty2014} and vendor libraries such as Vitis Vision~\cite{vitisvision} generate line-buffered pipelines automatically, and fully pipelined, line-buffered processing at one pixel per clock with deterministic latency also underlies recent low-latency line-segment detectors~\cite{jalilvand2026}. In streaming designs, hysteresis is typically omitted or confined to a local window.

\textbf{Exact hysteresis in hardware.} Hysteresis can be made exact in a single raster pass by recasting it as connected-component analysis: edge pixels receive provisional labels, equivalent labels are merged, and a component is retained if it contains a strong pixel~\cite{alnajjar2010,jang2019}, building on single-pass connected-component architectures for FPGAs~\cite{ma2008}, which now label 4K video streams in real time~\cite{kowalczyk2021}. These designs are exact and fast, but their label and equivalence storage and merge logic stand apart from the local operators of the rest of the pipeline. A recent HLS design follows chains with an explicit stack, as software does, at the cost of a large share of the block RAM~\cite{malik2025}.

\textbf{Raster-scan propagation.} Rosenfeld and Pfaltz showed that alternating forward and backward raster sweeps, in which each pixel sees values already updated within the same sweep, propagate information across an image in a few passes~\cite{rosenfeld1966}; the idea underlies distance transforms and multi-pass connected-component labelling. Hysteresis is a binary morphological reconstruction (the weak mask reconstructed from the strong seeds), and Vincent's sequential reconstruction algorithm computes reconstruction with precisely such alternating raster and anti-raster sweeps~\cite{vincent1993}. Anacona-Mosquera et al.\ implemented sequential and hybrid grey-level reconstruction in FPGA hardware, running the anti-raster sweep through the same module with the pixel order reversed~\cite{anacona2017,anacona2020}. We do not claim the sweeps themselves. Our contribution is how they are realised for Canny: inside the \emph{same} two-clock $3\times3$ hysteresis operator as the rest of the library, with a one-clock left-neighbour feedback and decisions written back into the window's row-above line buffer, so that each pass is a single streaming pass requiring no memory beyond the existing frame buffers.

\textbf{Fair comparison.} Following~\cite{lee2010,horowitz2014,kuon2007}, our CPU baseline includes both an optimised multi-threaded library and a single-threaded port of our own algorithm, and we report the cases in which the CPU is faster.

Table~\ref{tab:design} situates both paths among these approaches. None dominates on every criterion; our point is that the last three rows, two of them new, are all built from one set of operators.

\begin{table}[!htb]
\caption{Ways to implement Canny hysteresis on an FPGA. Pass counts are for our $128\times128$ pattern and photo tiles (Table~\ref{tab:inplace}).}
\label{tab:design}
\centering\footnotesize
\setlength{\tabcolsep}{4pt}
\resizebox{\columnwidth}{!}{%
\begin{tabular}{@{}lccc@{}}
\toprule
Approach & Exact & Extra hardware & Passes \\
\midrule
Local window~\cite{neoh2004} & no & line buffers & 1, streaming \\
Labelling~\cite{alnajjar2010,jang2019} & yes & label tables, merge logic & 1 \\
Two-buffer (baseline) & yes & 2 frame buffers & \structTwoBufMin--\structTwoBufMax{} \\
\textbf{In-place (ours)} & yes & 2 frame buffers + 2 feedbacks & \structInplaceMin--\structInplaceMax{} \\
\textbf{Fast path (ours)} & no$^{\dagger}$ & 1 line buffer per stage & 1, streaming \\
\bottomrule
\multicolumn{4}{@{}l@{}}{\scriptsize $^{\dagger}$Never reports a false edge (Proposition~\ref{prop:subset}).}
\end{tabular}}
\end{table}

\section{Architecture}\label{sec:arch}
\begin{figure*}[t]
\centering
\begin{tikzpicture}[
  font=\footnotesize, >=Stealth,
  op/.style={draw, rounded corners=2pt, fill=blue!8, minimum height=6mm, minimum width=12mm, align=center},
  win/.style={draw, fill=gray!12, minimum height=6mm, minimum width=7mm, align=center},
  mem/.style={draw, fill=orange!15, minimum height=6mm, minimum width=10mm, align=center},
  io/.style={draw, dashed, minimum height=6mm, align=center}]
\node[io] (cam) {OV7670\\RGB565};
\node[op, right=3mm of cam] (gray) {gray};
\node[win, right=3mm of gray] (w1) {W};
\node[op, right=1mm of w1] (g) {gauss};
\node[win, right=3mm of g] (w2) {W};
\node[op, right=1mm of w2] (s) {sobel};
\node[win, right=3mm of s] (w3) {W};
\node[op, right=1mm of w3] (n) {nms};
\node[op, right=3mm of n] (t) {thresh};
\node[win, right=3mm of t] (w4) {W};
\node[op, right=1mm of w4] (h) {hyst\\\scriptsize $\times K$};
\node[mem, right=3mm of h] (fb) {1-bit\\frame};
\node[io, right=3mm of fb] (vga) {VGA};
\foreach \a/\b in {cam/gray,gray/w1,w1/g,g/w2,w2/s,s/w3,w3/n,n/t,t/w4,w4/h,h/fb,fb/vga} \draw[->] (\a) -- (\b);
\node[anchor=west, font=\footnotesize\bfseries] at ([yshift=6mm]cam.north west) {Fast path: streaming, decision $\approx(3+K)$ lines after the pixel};
\node[io, below=19mm of cam] (uart) {UART\\tile};
\node[mem, right=3mm of uart] (bram) {BRAM frame\\buffers};
\node[win, right=6mm of bram] (we) {W};
\node[op, right=1mm of we, minimum width=34mm] (ops) {gauss $|$ sobel $|$ nms $|$ thresh $|$ hyst};
\node[mem, right=6mm of ops] (bram2) {BRAM frame\\buffers};
\node[io, right=3mm of bram2] (uart2) {UART};
\draw[->] (uart) -- (bram);
\draw[->] (bram) -- (we);
\draw[->] (we) -- (ops);
\draw[->] (ops) -- (bram2);
\draw[->] (bram2) -- (uart2);
\draw[->] (bram2.south) -- ++(0,-3mm) -| node[pos=0.25, below, font=\scriptsize] {one raster pass per stage; hysteresis passes repeat until nothing changes} (bram.south);
\draw[->, red!70!black, thick] ([xshift=6mm]ops.north) -- ++(0,3.5mm) -| node[pos=0.25, above, font=\scriptsize] {in-place feedback: left neighbour (1 clk), row above (line buffer)} (we.north);
\node[anchor=west, font=\footnotesize\bfseries] at ([yshift=11mm]uart.north west) {Exact path: frame in BRAM, true hysteresis};
\end{tikzpicture}
\caption{One operator library, two architectures. Blue blocks are the shared $3\times3$ operators, each with two clocks of latency; W is a line-buffered $3\times3$ window. The fast path chains the operators with a window between stages; the exact path runs them as raster passes over block-RAM frame buffers.}
\label{fig:arch}
\end{figure*}

Fig.~\ref{fig:arch} shows the two architectures and the operators they share. Both are assembled from the same five operators, each consuming a $3\times3$ window; the paths differ only in how the operators are connected and where the data are held.

\subsection{Operator library}
Every operator consumes a $3\times3$ window and produces its output exactly two clocks later; validity flags, border flags and addresses travel outside the operators in matched delay lines. \texttt{gauss} computes $[1\,2\,1;2\,4\,2;1\,2\,1]/16$ with shifts and adds. \texttt{sobel} outputs $\min(255,(|G_x|+|G_y|)\gg2)$ and quantises the direction to 0, 45, 90 or 135 degrees by comparing $256\,|G_y|$ with $106\,|G_x|$ and $618\,|G_x|$ (tangents of 22.5 and 67.5 degrees in Q8), so that no divider, square root or arctangent is required. \texttt{nms} retains the centre only if it is not smaller than both neighbours along the gradient. \texttt{thresh} classifies the centre as strong (255), weak (128) or suppressed (0) with run-time thresholds $T_L/T_H$. \texttt{hyst} promotes a weak centre that has a strong 8-neighbour. Every stage forces its one-pixel frame border to zero. The two paths use the operators identically except for a single input of \texttt{hyst}, \texttt{fb\_en}, which enables the in-place feedback of Section~\ref{sec:inplace}; the fast path ties it to zero.

\subsection{Exact path}
The exact path holds a $128\times128$ tile in eight block-RAM buffers (input, blur, magnitude, direction, NMS, classes and two hysteresis buffers) and executes each stage as one raster pass through a single shared window and the operator selected for that pass. Each pass zeroes its border, and this boundary effect propagates inward by one ring per stage; a 4-pixel halo therefore makes the core of every tile identical to whole-image processing for the first four stages. Hysteresis, in contrast, can cross tile boundaries, so tiled output differs from whole-image output on \seamPctLow\,\% (20/40) and \seamPctHigh\,\% (50/100) of the pixels of our $480\times720$ test image (model). Tiles enter and leave over a 2~Mbaud UART; a command byte returns any intermediate buffer, and a status command reports the cycle count, the number of hysteresis passes and the mode of the last run.

\subsection{In-place hysteresis}\label{sec:inplace}
In the original \emph{two-buffer} design, each hysteresis pass reads one buffer and writes the other, so a promotion becomes visible only in the following pass and a chain grows by a single pixel per pass. Reordering the updates is safe because of the following property. Let $S$ be the set of strong pixels after thresholding, $W$ the set of weak pixels, and $E\subseteq W$ the weak pixels joined to $S$ by an 8-connected path that runs through $W$; Canny's output is $S\cup E$.

\begin{proposition}[Order independence]\label{prop:order}
Consider any procedure that (i) promotes a weak pixel only when one of its 8-neighbours is already strong, and (ii) stops only when no weak pixel has a strong 8-neighbour. When it stops, the strong pixels are exactly $S\cup E$.
\end{proposition}
\begin{IEEEproof}
By (i), each promoted pixel is adjacent to a pixel that was strong before it, so by induction on the order of promotions every strong pixel lies in $S\cup E$. Conversely, suppose some $p\in E$ is still weak when the procedure stops. Walk along its path from $S$ to $p$; the first pixel on the path that is still weak has a strong predecessor, which contradicts (ii).
\end{IEEEproof}

Both the two-buffer scheme and the in-place scheme below satisfy (i) and (ii); they therefore produce identical edge maps and differ only in the number of passes. (Both also stop at a cap of 256 passes, which no tile in Table~\ref{tab:inplace} reaches.)

In an in-place pass, each pixel sees the decisions already made earlier in the same pass, and successive passes alternate between forward and reverse raster order, as in~\cite{rosenfeld1966,vincent1993}. Two feedback paths are required, and both are inexpensive in a raster pipeline:
\begin{itemize}
\item \textbf{Left neighbour.} The decision for pixel $x-1$ leaves the operator's second stage one clock before that of pixel $x$. Feeding this registered output back into the second stage (\texttt{fb\_en}) promotes a weak pixel whose left neighbour has just been promoted. The loop is one clock long, and the operator's latency remains two clocks.
\item \textbf{Row above.} The window keeps the two previous rows in line buffers. In feedback mode, the row-above buffer is written with the pass's own decisions rather than with its input, so each row sees the updated row above it. Each decision is written back about $W-6$ samples before it is next needed.
\end{itemize}
A reverse pass reads addresses $N-1,\dots,0$ and maps write addresses back through $N-1-a$. Because promotion and the zero border are symmetric under a 180-degree rotation, the window and operators need no change. Frames still alternate between the two hysteresis buffers, so no pass reads and writes the same RAM, and the termination rule (stop after a pass that promotes nothing; cap at 256 passes) is unchanged. Fig.~\ref{fig:inplace} shows both feedback paths.

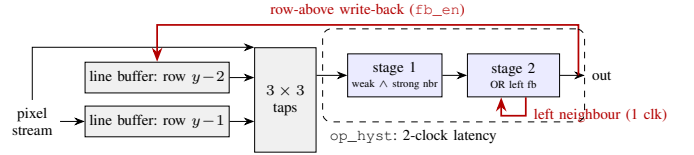
\begin{figure}[t]
\centering
\resizebox{\columnwidth}{!}{%
\begin{tikzpicture}[font=\scriptsize, >=Stealth,
  box/.style={draw, minimum height=5mm, align=center, inner sep=2pt},
  lb/.style={box, fill=gray!12, minimum width=20mm},
  st/.style={box, fill=blue!8, minimum width=15mm, minimum height=7mm}]
\node[lb] (lb1) {line buffer: row $y\!-\!1$};
\node[lb, above=2mm of lb1] (lb0) {line buffer: row $y\!-\!2$};
\node[box, right=4mm of lb1, minimum width=10mm, minimum height=17mm, fill=gray!12, yshift=3.5mm] (taps) {$3\times3$\\taps};
\node[st, right=5mm of taps, yshift=4mm] (s1) {stage 1\\[-1pt]{\tiny weak $\wedge$ strong nbr}};
\node[st, right=4mm of s1] (s2) {stage 2\\[-1pt]{\tiny OR left fb}};
\node[left=4mm of lb1, align=center] (in) {pixel\\stream};
\draw[->] (in) -- (lb1);
\draw[->] (in.north) |- (taps.north west);
\draw[->] (lb1) -- (lb1 -| taps.west);
\draw[->] (lb0) -- (lb0 -| taps.west);
\draw[->] (taps.east |- s1) -- (s1);
\draw[->] (s1) -- (s2);
\draw[->] (s2.east) -- ++(4mm,0) node[right] {out};
\draw[->, thick, red!70!black] ([xshift=2mm]s2.south) -- ++(0,-3mm) node[right] {left neighbour (1 clk)} -| ([xshift=-2mm]s2.south);
\draw[->, thick, red!70!black] ([xshift=3mm]s2.east) |- ([yshift=6mm]lb0.north) node[pos=0.75, above] {row-above write-back (\texttt{fb\_en})} -- (lb0.north);
\node[draw, dashed, rounded corners, fit=(s1)(s2), inner sep=4mm, label={[font=\scriptsize, anchor=north west]south west:\texttt{op\_hyst}: 2-clock latency}] {};
\end{tikzpicture}}
\caption{In-place hysteresis datapath. With \texttt{fb\_en}=1 the operator's own output is fed back into its second stage (left neighbour) and written into the row-above line buffer instead of the shifted input. Both loops are inside the existing pipeline; the operator's latency is unchanged.}
\label{fig:inplace}
\end{figure}

\subsection{Fast path}
The fast path chains the same operators with a streaming window between stages: camera, grayscale conversion, Gaussian, Sobel, NMS, threshold and $K$ local hysteresis stages, each a single plain (\texttt{fb\_en}$=0$) hysteresis iteration. The streaming window tolerates gaps in its input (the camera delivers a pixel every few clocks), supports widths that are not powers of two, and holds its two line buffers in one block RAM of $640\times2$ words. Because NMS needs magnitude and direction together, its window is 10 bits wide. Every stage zeroes the frame border, exactly as in the exact path. The output is a 1-bit $640\times480$ edge map in block RAM, scanned out to VGA at 60~Hz.

Each window delays the stream by one line and one pixel, so a decision emerges about $3+K$ lines after its pixel arrives. After the last pixel of a frame, $(3+K)(W+1)$ zero samples are injected at full clock rate during vertical blanking to flush the pipeline. The camera's pixel-clock domain is crossed through a Gray-code FIFO, and a frame controller resets the pipeline at the start of each frame and pads short frames. Grayscale is computed as $(77R+150G+29B)\gg8$ after expanding RGB565 to 8 bits per channel.

The accuracy of the fast path is bounded from one side:

\begin{proposition}[No false edges]\label{prop:subset}
Let $F_K$ be the edge set produced by the fast path with $K$ local stages, and let $E_k\subseteq E$ be the weak pixels whose shortest path from $S$ through $W$ has at most $k$ steps. Then $S\cup E_K\subseteq F_K\subseteq S\cup E$.
\end{proposition}
\begin{IEEEproof}
Every local stage satisfies condition (i) of Proposition~\ref{prop:order}, so the induction in that proof shows that every pixel the fast path promotes lies in $E$; weak pixels left at the end are discarded. This gives the upper bound. For the lower bound, each local stage reads only the previous stage's output and promotes every weak pixel that has a strong neighbour there, so after stage $k$ all pixels of $E_k$ are strong.
\end{IEEEproof}

The fast path therefore never reports an edge that the exact path rejects; it can only miss weak pixels lying more than $K$ steps from a strong pixel. Because its error is one-sided, its precision relative to the exact result is 100\,\% by construction, and recall is the appropriate measure of its accuracy.

\section{Verification Method}\label{sec:verif}
A one-cycle misalignment shifts the image by a pixel yet still yields a plausible-looking edge map. Every check is therefore a byte-level comparison against a software model, never a visual inspection.
\begin{enumerate}
\item \textbf{Golden models.} \texttt{canny\_model.py} mirrors the exact path operation for operation (shifts, truncation, tangent constants, zeroed borders, pass counting and cap) in both hysteresis modes; \texttt{fast\_model.py} does the same for the fast path for every $K$, including the RGB565-to-gray conversion. A C++ port of the exact model is byte-identical to it on all our test images and serves as the CPU baseline.
\item \textbf{RTL regression.} Icarus Verilog testbenches dump every intermediate stage. The exact path is exercised on \nTileCases{} tiles in both modes, with a lowered pass cap in both modes, and through the complete UART top level including the status command. The fast path is exercised on \nFastFrames{} frames of $640\times480$ for $K=1,2,4,8$ (\nFastRuns{} runs), including an RGB565 frame delivered with camera-like gaps. A continuous-integration job runs both regressions on every commit.
\item \textbf{System simulation.} The complete fast-path top level (capture, FIFO, frame controller, fast core and frame buffer) is simulated in Vivado's simulator against a behavioural OV7670 model that follows the datasheet's VGA RGB565 timing~\cite{ov7670}.
\item \textbf{Implementation.} Both designs are placed and routed in Vivado 2024.2 for the XC7A35T-1 and are required to close timing at 50~MHz.
\item \textbf{Hardware.} For the exact path, a link self-test first echoes a tile unchanged, and per-stage dumps then localise any mismatch. On the full $480\times720$ test image, 0 of 345,600 output pixels differed from the model on the board (two-buffer mode).
\end{enumerate}

\section{Results}\label{sec:results}
Results labelled RTL simulation come from cycle-accurate simulation of the Verilog, with cycle counts converted to time at 50~MHz. Fast-path latency comes from the system simulation with the camera model, resources and timing from Vivado place and route, and CPU times from measurement; the full-image result of Section~\ref{sec:verif} was obtained on the board.

\subsection{Exact path and in-place hysteresis}
Every stage of the exact path matches the model on all \nTileCases{} tiles in both modes, and the pass counts of RTL and model agree. Table~\ref{tab:inplace} quantifies the effect of in-place sweeps. On the pattern and photo tiles, the number of passes falls from \structTwoBufMin--\structTwoBufMax{} to \structInplaceMin--\structInplaceMax, and the hardest photo tile drops from \photoMidTwoBufMs{} to \photoMidInplaceMs~ms. Since every hysteresis pass costs the same number of cycles, the hysteresis step itself becomes \hystSpeedMin--\hystSpeedMax$\times$ faster on these tiles (\photoMidHystSpeed$\times$ on the hardest photo tile), and a complete tile, including the four fixed stages, becomes \tileSpeedMin--\tileSpeedMax$\times$ faster (last column of Table~\ref{tab:inplace}). Compute time is now dominated by the four fixed stages rather than by hysteresis. The serpentine tile, a single weak chain winding through the whole tile, is the worst case for two-buffer iteration (\snakeTwoBuf{} passes) yet needs only \snakeInplace{} in place: \hystSpeedWorst$\times$ fewer passes and a \tileSpeedWorst$\times$ faster tile. Gains are modest only where chains are short to begin with (noise, flat image). Fig.~\ref{fig:convergence} explains the high two-buffer pass count: most edges are found within the first few passes, but each long weak chain still grows by one pixel per pass, and this tail determines the number of passes. In-place sweeps eliminate the tail.

\begin{figure}[!htb]
\centering
\includegraphics[width=\columnwidth]{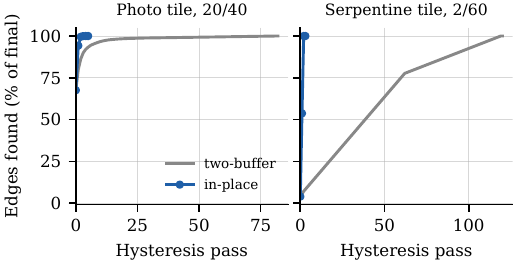}
\caption{Edges found after each hysteresis pass, as a fraction of the final edge set (model, bit-exact with the RTL). Both modes reach the same final set.}
\label{fig:convergence}
\end{figure} The same algorithmic change also pays off on a CPU: the single-threaded C++ port runs \refSpeedupLow$\times$ faster at 20/40 and \refSpeedupHigh$\times$ faster at 50/100 (Table~\ref{tab:cpu}).

\begin{table}[!htb]
\caption{Hysteresis passes and compute per $128\times128$ tile at 50~MHz, two-buffer vs.\ in-place (RTL simulation; outputs identical to the model in both modes)}
\label{tab:inplace}
\centering\footnotesize
\resizebox{\columnwidth}{!}{\input{generated/tab_inplace.tex}}
\end{table}

\subsection{Fast-path accuracy}
The fast-path edge map equals the model's for every frame and every $K$, and no run produced a single false edge. Table~\ref{tab:recall} and Fig.~\ref{fig:recall} report recall against exact hysteresis. Recall increases with $K$ and is higher at low thresholds, where strong pixels are denser and weak pixels lie closer to them: at 20/40 it rises from \recallPhotoKOne\,\% ($K=1$) to \recallPhotoKEight\,\% ($K=8$), and at 50/100 from \recallHighKOne\,\% to \recallHighKEight\,\%. Fig.~\ref{fig:example} shows what is missed: the long weak chains that only the exact path completes.

\begin{table}[!htb]
\caption{Fast-path recall of exact edge pixels (\%), RTL simulation. Precision is 100\,\% in every case.}
\label{tab:recall}
\centering\footnotesize
\input{generated/tab_recall.tex}
\end{table}

\begin{figure}[!htb]
\centering
\includegraphics[width=\columnwidth]{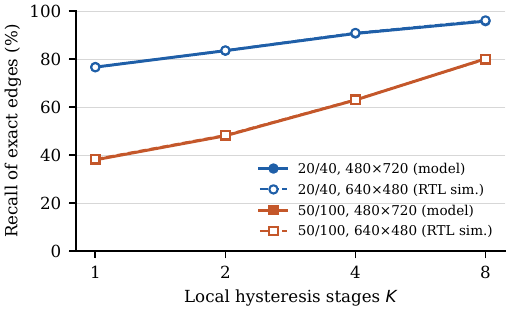}
\caption{Fast-path recall vs.\ number of local hysteresis stages $K$. Solid: whole $480\times720$ image (model). Dashed: $640\times480$ frame from the same photograph through the RTL; the two nearly coincide.}
\label{fig:recall}
\end{figure}

\begin{figure}[!htb]
\centering
\includegraphics[width=\columnwidth]{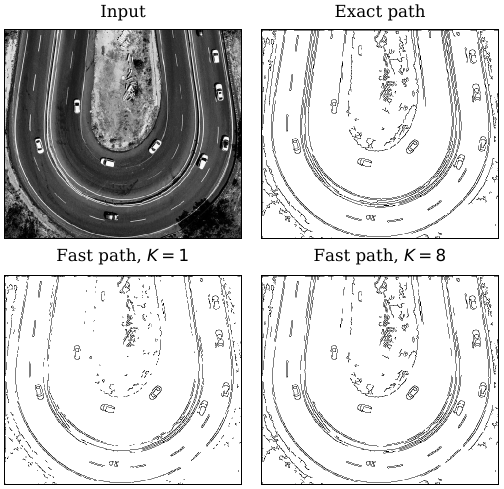}
\caption{Aerial photograph of a hairpin road ($500\times440$ crop, grayscale, thresholds 50/100), processed by the bit-exact models of both paths. The fast path reports only edges the exact path also reports. With $K=1$ the long, low-contrast lane markings around the bend break into fragments; with $K=8$ most of each curve is recovered.}
\label{fig:example}
\end{figure}

\subsection{Accuracy on a standard benchmark}
Recall against our own exact path measures how much the fast path loses, not how good either path is. We therefore also score both paths on the \bsdsN{} test images of the BSDS500 boundary benchmark~\cite{martin2004,arbelaez2011}, using the bit-exact models, so the scores hold for the RTL. Each predicted boundary pixel is matched one-to-one to human-annotated boundaries within 0.75\,\% of the image diagonal, the benchmark's tolerance; we re-implement this matching rule in Python rather than run the original benchmark code, and apply it identically to every detector. Thresholds are swept with $t_\mathrm{low}=t_\mathrm{high}/2$; ODS is the F-measure at the best common threshold and OIS at the best threshold per image. As Table~\ref{tab:bsds} shows, the exact path achieves ODS~F~=~\bsdsExactF, and OpenCV's Canny, which uses a different gradient scale and arithmetic, \bsdsCvF. The fast path scores between \bsdsFastMinF{} and \bsdsFastMaxF{} for every $K$ (\bsdsFastOneF{} with $K=1$, \bsdsFastEightF{} with $K=8$), retaining at least \bsdsFastMinPctOfExact\,\% of the exact path's score. The differences are not monotonic in $K$ because each detector is scored at its own best threshold, where a detector that drops weak chains can afford a lower $t_\mathrm{high}$. On this benchmark, the edge pixels missed by local hysteresis therefore do not measurably reduce agreement with human-drawn boundaries. Both paths score below OpenCV, whose gradient arithmetic differs from ours; matching its accuracy is not a goal of this design.

\begin{table}[!htb]
\caption{BSDS500 test set (\bsdsN{} images), boundary precision (P), recall (R) and F-measure (model, bit-exact with the RTL). $^{*}$OpenCV threshold on $|G_x|+|G_y|$, i.e.\ $4\times$ the models' scale.}
\label{tab:bsds}
\centering\footnotesize
\input{generated/tab_bsds.tex}
\end{table}

\subsection{Fast-path latency}
Table~\ref{tab:latency} reports the time from the camera delivering the second byte of the marker pixel (row 240, column 320) to its edge decision leaving the fast core, measured in the system simulation with the camera's datasheet timing (24~MHz pixel clock, one line every \lineUs~$\mu$s). The latency equals $3+K$ camera lines to within a few hundredths of a line, exactly as the architecture predicts: \latKOne~$\mu$s for $K=1$ and \latKEight~$\mu$s for $K=8$. Because the pipeline contains no data-dependent control, the latency is independent of image content.

\begin{table}[!htb]
\caption{Fast-path latency, camera pixel to edge decision (system simulation, OV7670 VGA timing)}
\label{tab:latency}
\centering\footnotesize
\input{generated/tab_latency.tex}
\end{table}

\subsection{Resources}
Table~\ref{tab:resources} lists post-route utilisation. Both designs close timing at 50~MHz with positive slack. The exact path is dominated by its eight tile buffers (\exactBram{} of \bramTotal{} BRAM36). The fast path at $K=8$, including camera capture, FIFO, frame buffer and VGA output, occupies \fastLut{} of \lutTotal{} LUTs, and each stage costs one $640\times2$-word line buffer. Vivado's post-route power estimate is \exactPowerEst~W for the exact path and \fastPowerEst~W for the fast path (default switching activity; an estimate, not a measurement).

\begin{table}[!htb]
\caption{Post-route utilisation on XC7A35T-1 at 50~MHz (Vivado 2024.2)}
\label{tab:resources}
\centering\footnotesize
\resizebox{\columnwidth}{!}{\input{generated/tab_resources.tex}}
\end{table}

\subsection{Comparison with a CPU}\label{sec:cpu}
Table~\ref{tab:cpu} gives single-image times on a \cpuModel{} laptop CPU (\cpuRuns{} timed runs after warm-up, $480\times720$ image, Linux subsystem on Windows, \texttt{g++ -O3 -march=native}; on mains power, minimum processor state 100\,\%, single-thread runs pinned to one core). OpenCV's Canny with 16 threads takes \cvAllMed~ms; on stored images, the CPU is clearly faster than our 50~MHz pipelines. The picture reverses when the image comes from a camera. A frame-based CPU program cannot start until the frame has arrived, which takes \framePeriodMs~ms at 30 frames per second: the first pixel of a frame waits that long before any processing begins, and only the last pixel experiences a latency close to the processing time. The fast path decides \emph{every} pixel \latKOneMs--\latKEightMs~ms after it arrives, with no operating system in the loop. For the first pixels of a frame, this is \soonerKEight--\soonerKOne$\times$ sooner than a frame-based program can even start (derived: \framePeriodMs~ms divided by the measured latency). This comparison is with frame-based software; a CPU or microcontroller program written to process the image line by line could also start early, and what the pipeline adds is a latency fixed by its structure rather than by scheduling. The exact path offers the complementary trade-off: full hysteresis, bit-exact with a reference model, at a fixed and analysable cost per tile.

\begin{table}[!htb]
\caption{CPU time per $480\times720$ image (\cpuModel)}
\label{tab:cpu}
\centering\footnotesize
\resizebox{\columnwidth}{!}{\input{generated/tab_cpu.tex}}
\end{table}

\section{Discussion and Limitations}
\textbf{Latency, not throughput.} At one pixel per clock and 50~MHz, neither path outperforms an optimised CPU on stored frames, and we do not claim otherwise. The fast path's value is a decision within a few lines of the pixel, at a latency fixed by the pipeline; the exact path's value is exact hysteresis that can be verified byte for byte.

\textbf{Two answers from one library.} The fast path delivers an early, conservative answer that never contains a false edge; the exact path completes the weak chains later. Because both paths share the same operators and the same models, the relation between their answers is known exactly: the fast result is always a subset of the exact one.

\textbf{Limitations.} The exact path processes $128\times128$ tiles because larger buffers do not fit in the device's block RAM; this causes the tile seams quantified above, which external memory would remove. Tiles travel over a UART that is far slower than the computation, so the exact path serves as a verification and refinement engine rather than a video-rate one. The fast-path results derive from RTL and system simulation with a datasheet camera model and from Vivado implementation, and the in-place mode of the exact path has been verified in RTL simulation; confirming both on the board, together with measured power, remains future work. The CPU results cover a single laptop processor. Finally, a single global threshold pair cannot suit both dense texture and soft edges, a limitation inherent to Canny itself.

\section{Conclusion}
A single library of hardwired $3\times3$ operators, each with a two-clock latency, serves both a streaming fast path and an exact frame-buffered path for Canny edge detection on a low-cost Artix-7. Built into that library, in-place forward and reverse sweeps reduce exact hysteresis from up to \structTwoBufMax{} to at most \structInplaceMax{} passes on photo and pattern tiles with identical output, accelerating the hysteresis step \hystSpeedMin--\hystSpeedMax$\times$ and complete tiles \tileSpeedMin--\tileSpeedMax$\times$ (\tileSpeedWorst$\times$ on a worst-case tile). The fast path commits each edge decision \latKOne--\latKEight~$\mu$s after the camera delivers the pixel, never reports a false edge, recovers up to \recallPhotoKEight\,\% of exact edges, and retains at least \bsdsFastMinPctOfExact\,\% of the exact path's boundary accuracy on BSDS500. For systems that must act on a camera stream before the frame is complete, this provides early, provably conservative answers and exact refinement from the same verified hardware. All results trace to scripts and raw logs in the accompanying repository.

\bibliographystyle{IEEEtran}
\bibliography{refs}

\end{document}

%% file: generated/numbers.tex
\newcommand{\bramTotal}{50}
\newcommand{\bsdsCvF}{0.508}
\newcommand{\bsdsExactF}{0.476}

\newcommand{\bsdsFastEightF}{0.470}

\newcommand{\bsdsFastMaxF}{0.477}
\newcommand{\bsdsFastMinF}{0.470}
\newcommand{\bsdsFastMinPctOfExact}{98.6}
\newcommand{\bsdsFastOneF}{0.477}

\newcommand{\bsdsN}{200}
\newcommand{\cpuModel}{12th Gen Intel Core i5-12500H}
\newcommand{\cpuRuns}{1000}
\newcommand{\cvAllMed}{0.603}

\newcommand{\exactBram}{29}

\newcommand{\exactPowerEst}{0.108}

\newcommand{\fastLut}{1455}
\newcommand{\fastPowerEst}{0.300}

\newcommand{\framePeriodMs}{33.3}
\newcommand{\hystSpeedMax}{21.7}
\newcommand{\hystSpeedMin}{6.6}
\newcommand{\hystSpeedWorst}{40}

\newcommand{\latKEight}{720.63}
\newcommand{\latKEightMs}{0.72}
\newcommand{\latKOne}{262.15}
\newcommand{\latKOneMs}{0.26}
\newcommand{\lineUs}{65.33}
\newcommand{\lutTotal}{20800}
\newcommand{\nFastFrames}{5}
\newcommand{\nFastRuns}{20}
\newcommand{\nTileCases}{8}
\newcommand{\photoMidHystSpeed}{16.4}

\newcommand{\photoMidInplaceMs}{2.97}

\newcommand{\photoMidTwoBufMs}{28.42}

\newcommand{\recallHighKEight}{79.9}
\newcommand{\recallHighKOne}{38.3}

\newcommand{\recallPhotoKEight}{96.1}
\newcommand{\recallPhotoKOne}{76.6}

\newcommand{\refSpeedupHigh}{4.4}
\newcommand{\refSpeedupLow}{2.4}

\newcommand{\seamPctHigh}{0.815}
\newcommand{\seamPctLow}{0.765}
\newcommand{\snakeInplace}{3}
\newcommand{\snakeTwoBuf}{120}
\newcommand{\soonerKEight}{46}
\newcommand{\soonerKOne}{127}
\newcommand{\structInplaceMax}{7}
\newcommand{\structInplaceMin}{3}
\newcommand{\structTwoBufMax}{82}
\newcommand{\structTwoBufMin}{38}
\newcommand{\tileSpeedMax}{9.9}
\newcommand{\tileSpeedMin}{4.5}
\newcommand{\tileSpeedWorst}{17.7}

\newcommand{\twoBufRealMax}{82}

%% file: generated/tab_inplace.tex
\begin{tabular}{@{}llrrrrr@{}}
\toprule
 & & \multicolumn{2}{c}{Hysteresis passes} & \multicolumn{3}{c}{Compute per tile (ms)} \\
\cmidrule(lr){3-4}\cmidrule(l){5-7}
Tile ($128\times128$) & $T_L/T_H$ & Two-buf. & In-place & Two-buf. & In-place & Speed-up \\
\midrule
Synthetic pattern & 25/60 & 38 & 3 & 13.88 & 2.31 & 6.0$\times$ \\
Synthetic pattern & 50/100 & 65 & 3 & 22.80 & 2.31 & 9.9$\times$ \\
Photo, top-left & 20/40 & 46 & 7 & 16.52 & 3.63 & 4.5$\times$ \\
Photo, middle & 20/40 & 82 & 5 & 28.42 & 2.97 & 9.6$\times$ \\
Photo, middle & 50/100 & 65 & 4 & 22.80 & 2.64 & 8.6$\times$ \\
Uniform noise & 20/40 & 5 & 4 & 2.97 & 2.64 & 1.1$\times$ \\
Flat white & 20/40 & 1 & 1 & 1.65 & 1.65 & 1.0$\times$ \\
Serpentine chain & 2/60 & 120 & 3 & 40.97 & 2.31 & 17.7$\times$ \\
\bottomrule
\end{tabular}

%% file: generated/tab_recall.tex
\begin{tabular}{@{}lrrrr@{}}
\toprule
Frame ($640\times480$) & $K=1$ & $K=2$ & $K=4$ & $K=8$ \\
\midrule
Photo, 20/40 & 76.6 & 83.6 & 90.9 & 96.1 \\
Photo, 50/100 & 38.3 & 48.3 & 63.1 & 79.9 \\
Camera RGB565, 20/40 & 77.7 & 84.5 & 91.6 & 96.6 \\
\bottomrule
\end{tabular}

%% file: generated/tab_bsds.tex
\begin{tabular}{@{}lrrrrr@{}}
\toprule
Detector & $t_\mathrm{high}$ & P & R & ODS F & OIS F \\
\midrule
Exact path (converged) & 64 & 0.360 & 0.704 & 0.476 & 0.510 \\
Fast path, $K=1$ & 50 & 0.360 & 0.706 & 0.477 & 0.512 \\
Fast path, $K=2$ & 64 & 0.385 & 0.608 & 0.472 & 0.507 \\
Fast path, $K=4$ & 64 & 0.371 & 0.639 & 0.470 & 0.504 \\
Fast path, $K=8$ & 64 & 0.363 & 0.668 & 0.470 & 0.504 \\
OpenCV Canny (L1) & 256$^{*}$ & 0.405 & 0.682 & 0.508 & 0.556 \\
\bottomrule
\end{tabular}

%% file: generated/tab_latency.tex
\begin{tabular}{@{}lrrr@{}}
\toprule
$K$ & Latency ($\mu$s) & Camera lines & Expected $3+K$ \\
\midrule
1 & 262.15 & 4.01 & 4 \\
2 & 327.65 & 5.01 & 5 \\
4 & 458.65 & 7.02 & 7 \\
8 & 720.63 & 11.03 & 11 \\
\bottomrule
\end{tabular}

%% file: generated/tab_resources.tex
\begin{tabular}{@{}lrrrrrr@{}}
\toprule
Design & LUT & FF & BRAM36 & BRAM18 & DSP & WNS (ns) \\
\midrule
Exact path (UART, $128\times128$ tiles) & 720 & 660 & 29 & 0 & 2 & 8.408 \\
Fast path, $K=8$ (camera to VGA) & 1455 & 2201 & 11 & 11 & 2 & 7.927 \\
\bottomrule
\end{tabular}

%% file: generated/tab_cpu.tex
\begin{tabular}{@{}lrrrrr@{}}
\toprule
 & & \multicolumn{2}{c}{20/40 (ms)} & \multicolumn{2}{c}{50/100 (ms)} \\
\cmidrule(lr){3-4}\cmidrule(l){5-6}
Implementation & Thr. & median & p99 & median & p99 \\
\midrule
C++ port, two-buffer & 1 & 19.49 & 25.32 & 28.01 & 33.72 \\
C++ port, in-place & 1 & 8.23 & 9.88 & 6.32 & 8.79 \\
OpenCV & 1 & 3.88 & 4.92 & 3.00 & 3.95 \\
OpenCV & 16 & 0.60 & 0.87 & 0.51 & 0.71 \\
\bottomrule
\end{tabular}